\documentclass[letterpaper,USenglish,cleveref]{lipics-v2021}

\usepackage{macros}

\title{Constructive Separations from Gate Elimination} 

\author{Marco Carmosino}{MIT-IBM Watson AI Lab, Cambridge, MA, USA}{mlc@ibm.com}{https://orcid.org/0009-0007-1118-1352}{}

\author{Ngu Dang}{Boston University, Boston, MA, USA}{ndang@bu.edu}{https://orcid.org/0009-0004-2774-2247}{}

\author{Tim Jackman}{Boston University, Boston, MA, USA}{tjackman@bu.edu}{https://orcid.org/0000-0002-2293-5670}{}

\authorrunning{M. Carmosino, N. Dang, and T. Jackman} 

\Copyright{Marco Carmosino, Ngu Dang, and Tim Jackman} 



\category{} 

\relatedversion{} 
\supplement{\supplementdetails[linktext={Machine-Readable Convergence Proofs},  swhid={Software Heritage Identifier}]{General Classification (e.g. Software, Dataset, Model, ...)}{https://feasible-math.org/} }

\acknowledgements{}

\EventEditors{John Q. Open and Joan R. Access}
\EventNoEds{2}
\EventLongTitle{42nd Conference on Very Important Topics (CVIT 2016)}
\EventShortTitle{CVIT 2016}
\EventAcronym{CVIT}
\EventYear{2016}
\EventDate{December 24--27, 2016}
\EventLocation{Little Whinging, United Kingdom}
\EventLogo{}
\SeriesVolume{42}
\ArticleNo{23}
\begin{document}
\nolinenumbers
\maketitle

\begin{abstract}
    Gate elimination is the primary technique for proving explicit lower bounds against general Boolean circuits, including Li and Yang's state-of-the-art $3.1n - o(n)$ bound for affine dispersers (STOC 2022). Every circuit lower bound is implicitly existential: every circuit that is too small to compute $f$ must err on some input. This raises a natural question: are these lower bounds \emph{constructive}? That is, can we efficiently produce such errors? Chen, Jin, Santhanam, and Williams showed that constructivity plays a central role in many longstanding open problems in complexity theory, and explicitly raised the question of which circuit lower bound techniques can be made constructive (FOCS 2021).

    We show that a variety of gate elimination arguments yield refuters -- efficient algorithms that, when given a circuit that is too small to compute a function $f$, produce an input on which the circuit errs. Our results range from elementary lower bounds for $\XOR$ and the multiplexer to more sophisticated arguments for affine dispersers. Underlying these results is a shift in perspective: gate elimination arguments \emph{are} algorithms. Each step either simplifies the circuit or reveals a violation of some structural or functional property, from which, with a little additional work, explicit counterexamples can be extracted. 
    
    We further strengthen the $\XOR$ result to handle circuits that \emph{match} the lower bound: given any DeMorgan circuit of size $3(n-1)$ that fails to compute $\XOR_n$, we can efficiently produce a counterexample. While refuters follow from the gate elimination arguments themselves, this refinement requires a complete characterization of the set of optimal circuits computing $\XOR$ -- a requirement rarely met by other explicit functions.
\end{abstract}

\newpage

\section{Introduction}
One of the most important open problems in complexity theory is to characterize the functions computed by polynomial-size Boolean circuits.
By a straightforward counting argument, most functions require huge circuits of size roughly $2^n / n^2$ .
But, after decades of research, it remains open to identify an \emph{explicit} Boolean function (in $\P$ or even $\NP$) that requires even \emph{super-linear} circuit size.

The state of the art is a $3.1n - o(n)$ lower bound for \emph{affine dispersers} (which can be computed in $\P$), proved by Li and Yang in 2021 \cite{Li022}.
They used \emph{gate elimination,} the best known technique for proving explicit and unconditional lower bounds against general Boolean circuits.
Proving a lower bound for $\func{f}$ via gate elimination works by setting the inputs of a circuit $\cC$ to carefully-chosen values and simplifying $C$ until we either (1) derive a contradiction to the assertion ``$\cC$ computes $\func{f}$'' or (2) shrink a \emph{complexity measure} of $\cC$.

Gate elimination has been used to push the research frontier in circuit lower bounds since 1974 \cite{Schnorr74}.
Yet we remain unable to identify an explicit function in $\NP$ that requires circuits of size $10n$.
Indeed, a certain formalization of gate elimination cannot prove lower bounds better than $5n$ \cite{DBLP:journals/jcss/GolovnevHKK18}.
However, the barrier does not cover every argument by gate elimination, including those that depend on the optimal structure of circuits computing $\func{f}$ or rely on special properties of the hard function $\func{f}$ to successful carry out induction.

Gate elimination arguments that evade this barrier have been used to obtain breakthrough circuit lower bounds \cite{DBLP:journals/tcs/Blum84,FindGHK2016}, establish ETH-hardness of the partial Minimum Circuit Size Problem ($\MCSP^\star$) \cite{DBLP:journals/siamcomp/Ilango24},
design algorithms for the $\XOR$-Simple Extension\footnote{The $\func{f}$-Simple Extension Problem was an important part in ETH-hardness of $\MCSP^\star$.}
Problem \cite{CarmosinoDJ2025}, and characterize the set of optimal circuits for a variety of key Boolean functions, including addition \cite{Redkin1973}.
What can be said about this family of arguments, if anything?
Towards a better understanding of gate elimination in general, this work shows that selected lower bounds for $\XOR$, the multiplexer, and affine dispersers can be made \emph{constructive.}

The statement ``$\func{f}$ does not have size $s(n)$ circuits'' is implicitly existential; every circuit that fails to compute $\func{f}$ must err on at least one input $x$.
Thus, every circuit lower bound for a specific function $\func{f}$ induces a total search problem: given as input a circuit $\cC$ that is too small (and therefore \emph{must} fail to compute $\func{f}$) print a bitstring on which $\cC$ and $\func{f}$ disagree.
Constructive lower bounds  play a surprisingly central role in complexity theory \cite{ChenJSW24}.

On the one hand, extremely difficult but widely-conjectured separations like $\P \neq \NP$ are \emph{automatically} constructive if they hold at all.
On the other hand, if certain \emph{weak} and \emph{simple} complexity lower bounds like ``Palindromes require super-linear time to decide on one-tape Turing machines'' can be made even $\P^\NP$-constructive, then $\P \neq \NP$!
Paraphrasing Chen, Jin, Santhanm and Williams: the intuition that ``easy to prove'' lower bounds can be made constructive is ``wildly inaccurate'' \cite{ChenJSW24}.
Their work studied complexity separations for uniform classes of algorithms, and concluded by suggesting that ``it would be interesting to examine which proof methods for circuit lower bounds can be made constructive''.

Because gate elimination is the best tool we have to prove lower bounds for \emph{unrestricted} Boolean circuits, it is a natural starting point.
Furthermore, arguments via gate elimination seem difficult to categorize as ``easy'' or ``hard''.
Globally, state-of-the-art lower bounds use deeply nested case analyses and carefully tailored complexity measures.
Locally, each individual case of these arguments employs totally elementary manipulation of easily-identified subcircuits.
We do not resolve the ``proof complexity'' of gate elimination here, but are able to show that selected arguments \emph{hide efficient circuit-analysis algorithms.}
 
\subsection{Our Results \& Contributions}
We consider circuit size over two different bases.
The \emph{DeMorgan} basis contains binary AND and OR gates and NOT gates.
The \emph{$\mathbb{B}_2$} basis contains every binary Boolean function.
The \emph{size} of a circuit $\C$ is the number of binary gates.  
Recall that the \emph{multiplexer} ($\MUX_n$, also known as the storage access function) is a Boolean function on $n + 2^n$ bits, where the first $n$ \emph{address bits} determine which of the $2^n$ \emph{data bits} to output (Definition \ref{def:mux}).
Towards extracting a constructive separation, in Section \ref{sec:lower-bound}, we give an elementary proof of

\begin{lemma}
  \label{lem:MUX-CktLB}
  Let $N = 2^n$.  $\MUX_n$ requires DeMorgan-size at least $2N + \log N - 2$ to compute.
\end{lemma}

Lemma \ref{lem:MUX-CktLB} is not tight; the best \emph{upper bound} for $\MUX_n$ is a DeMorgan circuit of size $2N + O(\sqrt{N})$ \cite{KleinP1980}.
Even so, it demonstrates constructivity in gate elimination.
We recall the definition of a constructive separation for non-uniform algorithms.

\begin{definition}[Refuters against circuits \cite{ChenJSW24}]
  Let $\func{f} : \{0,1\}^\star \to \{0,1\}$ be a Boolean function, and let $\Gamma$ denote a standard complexity class.
  A $\Gamma$-\emph{refuter} for $\func{f}$ against size $s(n)$ circuits is a $\Gamma$-algorithm $R$ that, given input $1^n$ and any circuit $\cC$ of size strictly less than $s(n)$, prints a string $x \in \{0,1\}^n$ such that $R(1^n, \cC) \neq \func{f}(x)$ for almost every $n$.
  We say there is a $\Gamma$-\emph{constructive separation} of $\func{f}$ from size $s(n)$ if there is a $\Gamma$-refuter for $\func{f}$ against $s(n)$-size circuits.
  When $\Gamma = \P$, we drop the prefix and just say \emph{refuter} and \emph{constructive separation}.
\end{definition}

Inspecting the proof of Lemma \ref{lem:MUX-CktLB}, in Section \ref{sec:mux-refuter}, we obtain

\begin{theorem}
There is a constructive separation of $\MUX$ from DeMorgan-size $2N + \log N - 2$.
\end{theorem}

The most advanced circuit lower bounds proved via gate elimination are for \emph{affine dispersers:} Boolean functions that are non-constant on every $d$-dimensional affine subspace of $\mathbb{F}_2^n$.
So we capture their constructivity by introducing the following class of search problems.

\begin{definition}
  Let $d : \mathbb{N} \to \mathbb{N}$ be an arbitrary function.  An \emph{affine refuter} for dimension $d$ against size $s(n)$ circuits is a $\P$-algorithm that, given input $1^n$ and any circuit $C$ of size strictly less than $s(n)$, prints the description of an affine subspace of dimension $d(n)$ on which $C$ is constant.
\end{definition}

Inspecting the elementary proof of $3n - o(n)$ lower bounds for affine dispersers \cite{DemenkovK11}, in Section \ref{sec:affine-refuter}, we obtain

\begin{theorem}
  There is an affine refuter for dimension $d$ against $\mathbb{B}_2$-size $3n - 4d(n)$.
\end{theorem}

\begin{corollary}
    Fix any explicit affine disperser $\func{f} \in \P$ of dimension $o(n)$.  There is a $\P^\NP$-constructive separation of $\func{f}$ from circuits of size $3n - o(n)$.
\end{corollary}

Denote by $\XOR_n$ the sum modulo 2 of $n$ bits. There is a generalization of Schnorr's lower bound that totally characterizes the set of optimal DeMorgan-circuits computing $\XOR_n$ \cite{CarmosinoDJ2025}. In Section \ref{sec:xor-refuter}, we first inspect Schnorr's classic argument to obtain

\begin{theorem}
    There is a constructive separation of $\XOR$ from DeMorgan-size $3(n-1)$.
\end{theorem}

But an argument that characterizes the set of optimal circuits should contain more powerful circuit analysis algorithms.
Indeed, we refine the refuter for $\XOR$ to solve the problem of efficiently \emph{checking} correctness of claimed optimal circuits for $\XOR$.
In particular, given a circuit $\cC$ of size exactly $3(n - 1)$ that purports to compute $\XOR_n$ but \emph{does not}, can we efficiently produce bitstrings $x$ such that $\cC(x) \neq \XOR_n(x)$?
Note that this is no longer a total search problem.
Even so, the task is feasible. In Section \ref{sec:improved-refuter}, we obtain

\begin{theorem}
    Let $\cC$ be any DeMorgan circuit on $n$ inputs with exactly $3(n-1)$ binary gates but does not compute $\XOR_n$. There is an efficient and deterministic algorithm that, given input $\cC$, prints an $n$-bit string $x$ such that $\cC(x) \neq \XOR_n(x)$.
\end{theorem}

\subsection{Related Work}

\subparagraph*{Feasible Mathematics and Bounded Arithmetic.} 
One way to obtain refuters is to prove complexity separations in \emph{bounded arithmetics:} weak fragments of Peano Arithmetic that enjoy tight connections to efficient algorithms.  In particular, these theories have \emph{efficient witnessing} where proofs of existential statement can be automatically converted into efficient algorithms that print the objects asserted to exist.  Some of the lower bounds for \emph{restricted} circuit classes, like $\AC^0$, have been formalized in bounded arithmetic and so have efficient refuters \cite{DBLP:journals/apal/MullerP20}.  To the best of our knowledge, no circuit lower bound via gate elimination has been formalized in bounded arithmetic.

Grosser and Carmosino showed that many of the dramatic implications that come from making simple and weak lower bounds constructive in the \emph{algorithmic} sense still hold (in slightly weaker forms) when we consider proofs in \emph{weak fragments of Peano Arithmetic} closely associated with computational complexity classes, like Cook's $\lang{PV}$ (which is connected to deterministic polynomial time) \cite{GrosserC25}.

\subparagraph*{Lower Bounds for $\MUX$.}
An asymptotically tight circuit lower bound for $\MUX$ over the DeMorgan basis is stated by \cite{LozhkinK2021}, but we cannot locate a published proof.  Tight lower bounds are required to characterize the set of optimal $\MUX$ circuits and derive an efficient checker. 
\subparagraph*{Checking.}
Our checker for $\XOR$ is like a white-box and deterministic version of a \emph{program checker} introduced by Blum and Kannan in the 90s \cite{BlumK95}. In particular, 
a program checker is a \emph{probabilistic} algorithm that--- given a \emph{deterministic} program $P$, an input instance $I$--- certifies with high probability whether the output produced by $P(I)$ is correct or instead declares the program $P$ to be buggy. A central conceptual contribution of this work is the insight that the structural properties of the underlying problem can be leveraged to make checking efficient: by querying the program on carefully chosen related instances, one can test consistency conditions that any correct computation must satisfy.

Blum and Kannan provided concrete checkers for some explicit polynomial time computable problems such as sorting, computing the matrix rank, the greatest common divisor, and for several group-theoretic problems.
Our results extends this line of research to the Boolean circuits' regime by showing how one can utilize the known characterization of optimal $\XOR$-circuits to efficiently and \emph{deterministically} check whether a given circuit $\cC$ (of size $3n - 3$) correctly computes $\XOR_n$.

\subsection{Proof Techniques}

\subsubsection{Elementary $\MUX$ Lower Bound}
Paul’s classical argument for a $2(N - 1)$ $\mathbb{B}_2$ lower bound for $\MUX_n$ only restricts individual \emph{data bits} \cite{Paul1975}. As a result, the restricted circuit must still multiplex the remaining data bits. To ensure that at least two gates are always eliminated, Paul's substitutes \emph{functions} of variables, rather than constants, in order to fix $\oplus$-type gates. As these substitutions may make certain address bits degenerate in the restricted function, the proof cannot manipulate the address bits to eliminate more gates. In Section \ref{sec:lower-bound} however, we show that the argument becomes significantly simpler in the DeMorgan basis. In $\mathbb{D}$, constant substitutions suffice to eliminate two gates per data bit. After substituting all of the data bits which agree on a certain address bit $a_i$, the underlying function is guaranteed to still depend on $a_i$. We can therefore substitute $a_i$ to eliminate even more gates. This accounts for an extra $\log N$ factor in the $2N + \log N - 2$ lower bound established in Lemma \ref{lem:MUX-CktLB}. 

\subsubsection{Efficient Refuters}
Gate elimination often involves arguments that enough specific simplifications can be applied after certain substitutions. If not, the circuit must violate a structural or functional property of its underlying function. After simplification, the circuit computes a smaller instance of the function (or a smaller instance in the function class) being considered and the lower bound follows inductively. 

In Section \ref{sec:refuters}, view proofs via gate elimination as \emph{structured algorithms}. We run the arguments on circuits which \emph{are} to small to compute the function by interpreting the case analysis as a branching algorithm. As the circuits do not meet the lower bound, this process \emph{must} either fail at some step or efficiently reduce the circuit to constant size. If the algorithm fails, then we can extract a counterexample from the violated structural or functional property. If the algorithm reaches a constant-size restriction, the refuter can brute force over all remaining inputs and find one on which it errs. The algorithms for $\MUX$ and $\XOR$ in Sections \ref{sec:mux-refuter} and Appendix \ref{sec:xor-refuter} respectively follow this structure exactly and are able to extract counterexamples efficiently. The affine refuter of Section \ref{sec:affine-refuter} also mirrors this structure and produce affine subspaces on which the circuits are constant.

\subsubsection{Efficient Optimal $\XOR$-circuits Checker}
While gate elimination arguments naturally extend to refuters they do not as readily produce checkers -- algorithms which find errors for incorrect circuits whose size \emph{meets} the lower bound. Generally, the issue with naively running a refuter as checker is that a ``successful'' substitution by the checker (i.e., one that succeeds at reducing the complexity of the circuit by the minimal amount) may produce a circuit which \emph{does} correctly compute the corresponding restricted function. The original circuit may have only erred on inputs where the variable was set according to a different value and therefore a checker will then never be able to find an error. 

In Section \ref{sec:xor-checker}, we show that it is sometimes possible to circumvent this by extending a refuter for $\XOR$ into a checker. To do so, we efficiently \emph{detect} when the reduced circuit is correct, so that the algorithm can instead make the opposite constant substitution. We leverage the characterization of optimal DeMorgan $\XOR$ circuits as binary trees of $\XOR_2$ and $\neg \XOR_2$ widgets from \cite{CarmosinoDJ2025}. As the converse of this theorem also holds (i.e., binary trees of $\XOR_2$ and $\neg \XOR_2$ widgets compute $\XOR$ (or $\neg \XOR$)), the checker (Algorithm \ref{alg:xor-checker}) can detect whether the simplified circuit computes $\XOR$ in polynomial time and backtrack one step if necessary. Afterwords, all that is required is to argue that the opposite substitution also eliminates enough gates or otherwise violates some property of $\XOR$.

\subsection{Future Directions}

These constructivity results for $\XOR$, $\MUX$, and an \emph{elementary} argument about affine dispersers immediately raise further questions: what about state-of-the-art circuit lower bounds?  Do the more sophisticated complexity measures impede refutation by efficient algorithms?  Are there dramatic consequences of making those arguments constructive? We also ask whether our refuter for $\MUX$ and the affine refuter can be extended to a checker and improved to a constructive separation respectively. Are there any consequences if there constructive separations of affine functions $\func{f}$ of dimension $d$ from size $3n - 4d(n)$?

Finally, it would be interesting to determine if gate elimination can be formalized in bounded arithmetic.
This would establish that not only is the refuter efficient, but also that the proof of correctness for the refuter is feasible. We conjecture that all the arguments from this paper can be formalized in $\mathsf{PV}$, which seems to contain enough linear algebra to argue about affine dispersers.


\section{Preliminaries}

\subsection{General Notation}
We write $[n]$ to represent the set $\{1, 2, \ldots, n\}$. We write addition modulo 2 over the field $\F_2 = \{0,1\}$ explicitly using $\oplus$. Vectors are written with an arrow e.g., $\vec{v} \in \F_2^n$, and addition between vectors $\vec{v} + \vec{u}$ denotes component-wise addition modulo 2. We write $\vec{e_i}$ for the $i^{\text{th}}$ standard basis vector, i.e, the vector whose $i$-th coordinate is $1$ and whose other coordinates are $0$. The all zero vector is $\vec{0}$, the $n \times n$ identity matrix is $\mathbf{I}_n$, and the $n \times n$ all-zeroes matrix is $\mathbf{0}_n$. For a set $I \subseteq [n]$, we write $\ind{I} \in \F_2^{1 \times n}$ for the row indicator vector of $I$, i.e., the row vector whose $i$-th coordinate is $1$ if and only if $i \in I$.

\subsection{Boolean Circuits}

We assume familiarity with the basics of Boolean circuit complexity; see \cite{Jukna} for a standard reference. We consider unrestricted circuits over the \emph{DeMorgan} basis, $\mathbb{D} =\{\land, \lor, \neg\}$, and the $\mathbb{B}_2$ basis consisting of all Boolean functions of fan-in at most $2$. As in \cite{DemenkovK11}, we divide the $16$ binary Boolean functions $\func{f}(x,y)$ of $\mathbb{B}_2$ into three categories: six degenerate functions (i.e., $0, 1, x, \neg x$, $y$, $\neg y)$, eight $\land$-type (i.e., $x \land y, \neg x \land y, \ldots, \neg(\neg x\land \neg y))$, and two $\oplus$-type (i.e., $x \oplus y$ and $\neg (x \oplus y)$). Observe that an optimal circuit computing a function over $\mathbb{B}_2$ never contains degenerate gates unless they're computing the degenerate functions themselves. 

The size of a circuit $\cC$, denoted $\sigma(\cC)$, is the number of binary gates in $\cC$. We denote by $\eta(\cC)$ the number of input variables in $\cC$ whose fanout is at least one. The circuit complexity of a function $\func{f} : \F_2^n \to \F_2$ over a basis $\mathbb{B}$, denoted $\func{CC}_\mathbb{B}$, is the minimum size of any $\mathbb{B}$-circuit computing $\func{f}$. When working with circuits over the DeMorgan basis, we will often write $(\neg)\alpha$ to refer to either a binary gate $\alpha$ or a $\neg$ gate reading $\alpha$.

\subsection{Gate Elimination}
Lower bounds for general Boolean circuits use the  \emph{gate elimination} method: first, fix some inputs to some value, then repeatedly simplify the resulting circuit. Simple proofs via gate elimination substitute constants and then simplify according to basic Boolean identities. Each simplification either replaces a gate by a constant or merges it into one of its inputs. We categorize these identities into the following three types:
\begin{itemize}
    \item \emph{Fixing} rules: directly turn a gate into the constant $0$ or $1$
    (e.g.\ $0\land\gamma \to 0$).
    \item \emph{Passing} rules: replace a gate by one of its inputs
    (e.g.\ $1\land\gamma \to \gamma$), letting that input inherit the outgoing
    wires of the gate.
    \item \emph{Trivial simplifications}: eliminate a literal and its negation
    (e.g.\ $\gamma \land \lnot\gamma \to 0$)
    or remove duplicate inputs
    (e.g.\ $\gamma \land \gamma \to \gamma$) and double negations.
\end{itemize}

We apply these simplifications repeatedly until we can no longer. We call such a circuit, in which no further identities can be applied, \emph{normalized} or \emph{in normal form}. Note that when working in $\mathbb{B}_2$, no single constant substitution can fix an $\oplus$-type gate. Hence, more advanced gate elimination arguments substitute functions. Examples of these substitutions are given in Figure \ref{fig:gate-elim-example}. While simplifying, gate elimination arguments track certain complexity measures; basic arguments simply track circuit size, thereby concluding that the original circuit had at least as many gates as were removed. More advanced gate elimination arguments, like \cite{DemenkovK11}, track additional quantities, such as $\eta(\cC)$.

\begin{figure}[t]
    \centering
    \captionsetup[subfigure]{justification=centering}
    \captionsetup[subfigure]{labelformat=empty}
    \begin{subfigure}{\textwidth}
        \centering
        \includegraphics[]{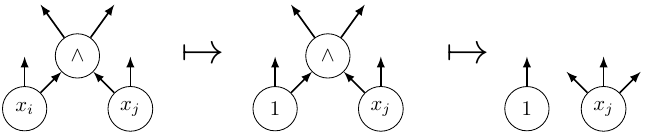}
        \caption{(a) A constant substitution and accompanying passing rule}
    \end{subfigure}
    \begin{subfigure}{\textwidth}
    \centering
        \includegraphics[]{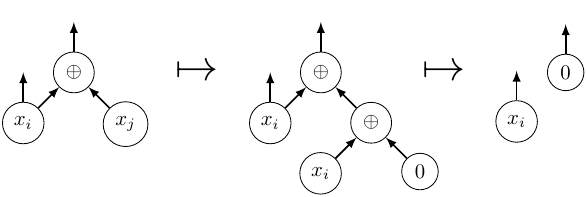}
        \caption{\hspace{-0.5cm}(b) A linear substitution and accompanying fixing rule}
    \end{subfigure}
    \caption{Two examples of simplifications in $\mathbb{D}$ and $\mathbb{B}_2$. In (a), $x_i$ is substituted for a constant which simplifies $\land$, passing its out wires to $x_j$. In (b), $x_j$ is substituted with the linear equation $x_j \gets x_i \oplus 0$. The top $\oplus$ gate therefore computes $x_i \oplus (x_i \oplus 0) = 0$ and thus becomes fixed.}
    \label{fig:gate-elim-example}
    \end{figure}

\subsection{Specific Functions and Function Classes}
We consider the following functions and function class.

\begin{definition}[$\XOR$]
\label{def:xor}
    For $\vec{x} \in F_2^n$, we define the $\XOR$ function on $n$ bits ($\XOR_n$) as
    \(
        \XOR_n(\vec{x}) = \bigoplus_{i \in [n]} x_i.
    \)
    For a set $I \subseteq [n]$, we define $\XOR_I(\vec{x}) = \bigoplus_{i \in I}x_i.$ 
\end{definition}
We use $(\lnot)\XOR$ to denote the $\XOR$ function itself or its negation $\lnot \XOR$.

\begin{definition}[$\MUX$]
\label{def:mux}
    We define the multiplexer $\MUX_n: \F_2^{n} \times \F_2^{2^n} \to \F_2$ as follows. For all $(\vec{a}, \vec{x})$ where $\vec{a} \in \F_2^n, \vec{x} \in \F_2^{2^n}$, we have
    $
        \MUX_{n}(\vec{a}, \vec{x}) = \vec{x}_{(\vec{a})}
    $,
    where $\vec{x}_{(\vec{a})}$ where $(\vec{a})$ is $\vec{a}$ interpreted as a base-2 number plus $1$, i.e., the $((\vec{a}) + 1)$-th bit of $\vec{x}$.
\end{definition}

\begin{definition}[Affine Subspace and Dispersers]
        A set $S \subseteq \F_2^n$ is an \emph{affine subspace of dimension $d$} if there exists a full column rank matrix $A \in \F_2^{n \times d}$ and vector $\vec{a} \in \F_2^d$ such that $S = \{A\vec{x} + \vec{a} \mid \vec{x} \in \F_2^n\}$. A function $\func{f} : \F_2^n \to \F_2$ is an \emph{affine disperser of dimension $d$} if $\func{f}$ is not constant on any affine subspace of dimension at least $d$.
\end{definition}

\section{An Elementary Lower Bound for $\MUX$}
\label{sec:lower-bound}
Paul's $2N-2$ lower bound for $\MUX_n$ in the $\mathbb{B}_2$ basis \cite{Paul1975} propagates to the DeMorgan basis. However, in the DeMorgan basis, the argument becomes simpler (and tighter) as constant substitutions suffice to restrict data bits and therefore the argument can take advantage of address bits. Substituting for an address bit makes half of the data bits irrelevant, allowing us to substitute irrelevant data bits to eliminate as many gates as possible. 

To streamline the proof of correctness for our refuter (Theorem \ref{thm:refuter-our-lb}), we present the proof of Theorem \ref{thm:improved-MUX-lb} in a labeled, case-based format.  The structure of the proof aligns closely with the structure of the refuter; labeling allows us to reference corresponding steps of the proof when analyzing the refuter. Cases will correspond to branches of our program. When Cases are proofs by contradiction, we label their conclusion as Assertions.

\begin{theorem}
\label{thm:improved-MUX-lb}
    In the DeMorgan basis, $\func{CC}(\MUX_n) \geq 2N + \log N - 2$ where $N = 2^n$.
\end{theorem}

\begin{proof}
    We will show that $\func{CC}_{\mathbb{D}}(\MUX_n) \geq 2 \cdot 2^n + n - 2$. We exploit the following observation about $\MUX_{n}$: restricting $a_1$ removes dependence on the half of the data bits. Hence, we can restrict those degenerate data bits to remove as many possible gates as possible. Afterwards, by downward self-reducibility, the circuit is a multiplexer on the remaining data bits and we can apply induction. 
    
    \subparagraph*{Base Case.} Observe that $\MUX_1$ requires at least $3 = 2 \cdot 2^1 + 1 - 2$ costly gates. $\MUX_1$ depends on all of its inputs: $a_1, x_1, x_2$ and thus there are at least two binary gates which read them. Furthermore, none of these inputs can feed the output as otherwise, substituting them to fix the output makes the circuit constant, but $\MUX_1\hook_{v \gets b}$ is not constant for any variable $v$ and constant $b$. Therefore there must be at least one more gate which reads the gates that read $a_1, x_1, x_2$.

    \subparagraph*{Inductive Case.}
    Let $\cC$ be a circuit computing $\MUX_n$ for $n > 1$. Let $D = \{x_i \mid i > 2^{n-1}\}$, i.e., the data inputs whose first address bit is $1$. We will  find a substitution for each $x_i$ in $D$ (i.e., exactly half of the data inputs of $\cC$) so that two binary gates are eliminated. 
    
    \begin{description}
    \item[Loop over $D$.] Fix $x_i \in D$. We will substitute for $x_i$ to eliminate at least $2$ binary gates.
    \begin{description}
        \item [Case 1: Data degeneracy.] Suppose $\cC$ does not read $x_i$. Then, $\cC$ does not depend on $x_i$ while $\MUX_n$ does. Hence, $\cC$ cannot compute $\MUX_n$.
        \item [Assertion 1.] $\cC$ reads $x_i \in D$ at least once.

    \item [Work: Capture Minimal Gate.] Let $\alpha_i$ be any gate reading $x_i$.
    
    \item [Case 2: Short-Circuit at $\alpha_i$.] Suppose $(\neg)\alpha_i$ is the output gate. Setting $x_i \leftarrow b$ which eliminates $\alpha_i$ via a fixing rule makes the circuit constant. However, $\MUX$, restricted with any set of substitutions to $D$ is not constant (as the complement of $D$ is not empty, and those data bits could be selected by the address).

    \item [Assertion 2.] The output of the circuit is not $(\neg)\alpha_i$.

    \item [Work: Eliminate Gates via Data Bit Restrictions.] Let $\beta_i$ be another binary gate reading $\alpha_i$. Setting $x_i$ to fix $\alpha_i$ eliminates at least $\alpha_i$ and $\beta_i$. Since $|D| = 2^{n - 1}$, processing all data bits in $D$ eliminates at least $2 \cdot 2^{n-1} = 2^n$ gates.
    \end{description}
    
    \item [Case 3: Address degeneracy.] After substituting all $2^{n-1}$ variables in $D$, suppose that the resulting circuit does not read $a_1$. Then the circuit output is independent of $a_1$, whereas the restricted multiplexer still depends on $a_1$. In particular, it determines whether the selected data bit lies in the free lower or fixed upper half. Hence, $\cC$ cannot compute $\MUX_n$.

    \item [Assertion 3.] After substitutions of the data variables in $D$, the circuit reads $a_1$ at least once.

    \item [Work: Eliminate One More Gate.] Substitute $a_1 \gets 0$ and simplify. The resulting circuit has lost at least one more binary gate. 
    
    \item [Inductive Hypothesis.] The simplified circuit now is a multiplexer on the remaining $2^{n-1}$ inputs which are addressed by $\{a_2, \ldots a_n\}$. Thus, we obtain
    \begin{align*}
        \sigma(\cC) &\geq \sigma(\MUX_{n-1}) + 2\cdot 2^{n-1} + 1 \\ &\geq (2 \cdot 2^{n-1} + (n-1) - 2) + 2\cdot 2^{n-1} + 1 \\ &= 2 \cdot 2^n + n - 2\\ &= 2N + \log N - 2. \qedhere
    \end{align*}
    \end{description}
\end{proof}

\section{Efficient Refuters from Gate Elimination}
\label{sec:refuters}

\subsection{A Refuter for $\MUX$}
\label{sec:mux-refuter}

We now show that our elementary lower bound for $\MUX$ is constructive by exhibiting an efficient refuter against circuits that violate our lower bound (Algorithm \ref{algo:mux-refuter-alt}).

\begin{algorithm}[t]
  \caption{Refuter for $\MUX_n$ for circuits of size less than $2N + \log N - 2$}
  \label{algo:mux-refuter-alt}
  \begin{algorithmic}[1]
    \Require $\cC$ is a normalized circuit, $m \in [n]$, $\vec{d} \in \F_2^N$ such that $\mu(\cC) < 2 \cdot 2^m + m - 2$ and $\cC$ only reads from $\{a_{n-m+1}, \ldots a_n\}$ and $\{x_i \mid \forall i\in [n], i < 2^m\}$
    \Ensure $\vec{w}^A\in \F_2^{\log N}, \vec{w}^X \in \F_2^N$ such that $\cC(\vec{w}^A, \vec{w}^X) \neq \vec{w}^X_{\Call{int}{\vec{w}^A}}$ and $\forall i \geq m, \vec{w}^A_i = 0$ and $\vec{w}^X_j = \vec{d}_j$ for all $j > 2^m$
    \Procedure{MUX-REFUTER}{$\cC, m, \vec{d}$}
        \If{$m = 1$}
            \LComment{Base Case: $\cC$ errs on at least one of these $4$ assignments as it is too small}
            \State \textbf{compute} $\cC$ on $(\Call{bin}{1}, \vec{d}^X), (\Call{bin}{1}, \vec{d}^X + \vec{e}_1), (\Call{bin}{2}, \vec{d}^X),$ and $(\Call{bin}{2}, \vec{d}^X + \vec{e}_2)$
            \State \Return an input on which $\cC$ does not agree with $\MUX$
        \EndIf
        \State $D \gets \{x_i \mid 2^{m-1} < i \leq 2^m\}$

        \For{each $x_i \in D$}
            \If{$\cC$ does not read $x_{i}$}
                \LComment{Case 1: $\cC$ is degenerate with respect to $x_{i}$ but the restricted $\MUX$ is not}
                \State \textbf{compute} $\cC$ on $(\Call{bin}{i}, \vec{d})$ and $(\Call{bin}{i}, \vec{d} + \vec{e}_i)$
                \State \Return whichever input $\cC$ does not agree with $\MUX$ on
            \EndIf
            \State $\alpha \gets$ topological minimal gate reading $x_{i}$
            \State $\vec{d}_{i} \gets$ constant whose substitution for $x_{i}$ fixes $\alpha$
            \If{$(\neg) \alpha$ is the output gate of $\cC$}
                \LComment{Case 2: The restricted circuit with $x_i \gets \vec{d}_i$ is constant but $\MUX$ is not}
                \State \textbf{compute} $\cC$ on $(\Call{bin}{2^{m-1}}, \vec{d})$ and $(\Call{bin}{2^{m-1}}, \vec{d} + \vec{e}_{2^{m-1}})$
                \State \Return whichever input $\cC$ does not agree with $\MUX$ on
            \EndIf
            \State $\cC \gets \cC$ with $\vec{d}_{i}$ substituted for $x_{i}$ and then simplified
        \EndFor

        \If{$\cC$ does not read $a_{n-m+1}$}
            \LComment{Case 3: $\cC$ is degenerate with respect to $a_{n-m+1}$ after setting all $x \in D$}
            \State $d_{1} \gets d_{2^{m-1}+1} + 1$
            \State \textbf{compute} $\cC$ on $(\Call{bin}{1}, \vec{d})$ and $(\Call{bin}{2^{m-1}+1}, \vec{d})$
            \State \Return whichever input $\cC$ does not agree with $\MUX$ on
        \EndIf
        \State $\cC \gets \cC$ with $a_{n-m+1}$ substituted by $0$ and simplified
        \State \Return \Call{MUX-REFUTER}{$\cC, m-1, \vec{d}$}
    \EndProcedure
  \end{algorithmic}
\end{algorithm}

\begin{theorem}
  \label{thm:refuter-our-lb}
    Let $C$ be a DeMorgan circuit on $n + N$ inputs, where $N = 2^n$, with size $s < 2N +\log N - 2$ that claims compute $\MUX_n$. We can construct an input $(\alpha, \rho)$ with $\alpha \in \{0, 1\}^n$, $\rho \in \{0, 1\}^{N}$, and $N = 2^n$, such that $C(\alpha, \rho) \neq \MUX_n(\alpha, \rho)$ in polynomial time with respect to $N$.
\end{theorem}

\begin{proof}
    To begin, we clarify the notation used in the algorithm and its inputs. The procedure \Call{MUX-REFUTER}{$\cC, m, \vec{d}$} operates on a restricted instance of the multiplexer where the least significant $m$ address bits remain unfixed. Concretely, the live address variables are $\{a_{n-m+1}, \dots, a_n\}$, while the remaining address bits have already been fixed to $0$. The vector $\vec{d} \in \F_2^N$ represents an assignment to the data inputs, where the entries $\vec{d}_i$ for $i \leq 2^m$ correspond to the currently live data inputs, while entries for $i > 2^m$ are fixed and will not be modified by the algorithm. 
    
    For an index $i \in [2^m]$, $\Call{bin}{i}$ denotes the binary encoding of $i$ on the $m$ live address bits, i.e., the assignment to $\{a_{n-m+1}, \dots, a_n\}$ that selects the $i$-th live data input. Conversely, for an address vector $\vec{w}^A$, $\Call{int}{\vec{w}^A}$ denotes the integer index selected by that address, so that the multiplexer outputs the data bit $\vec{w}^X_{\Call{int}{\vec{w}^A}}$. For any $i \leq 2^m$, we have $\MUX(\Call{bin}{i}, \vec{d}) = \vec{d}_i.$

    We argue that Algorithm~\ref{algo:mux-refuter-alt} correctly outputs an error by induction on $m \in [n]$, the number of current unfixed address bits. In the base case, when $m = 1$, the algorithm brute forces over the four relevant assignments and returns one on which $\cC$ disagrees with $\MUX$. This is correct by the \textbf{Base Case} of Theorem~\ref{thm:improved-MUX-lb}, since any circuit with fewer than $2\cdot 2^1 + 1 - 2 = 3$ costly gates cannot compute $\MUX_1$.
    For the inductive case, fix some $1 < m \leq n$ and assume that \textsc{MUX-REFUTER} correctly outputs an error on all inputs corresponding to addresses of length $m - 1$. Let 
    $D = \{x_i \mid 2^{m-1}< i \leq 2^m\}.$
    The subsequent branches of the algorithm correspond to the cases identified in the proof of Theorem~\ref{thm:improved-MUX-lb}. If the algorithm does not enter one of these branches, the corresponding assertions from the lower bound proof hold for $\cC$.
    
    \begin{description}
        \item[Branch at step 8] corresponds to Case 1 of Theorem~\ref{thm:improved-MUX-lb}. 
        Suppose the algorithm finds some $x_i\in D$ that $\cC$ does not read. The circuit $\cC$ does not depend on $x_i$, but the current restricted multiplexer does. The algorithm considers the two inputs $(\Call{bin}{i},\vec{d})$ and $(\Call{bin}{i},\vec{d}+\vec{e}_i)$. The circuit outputs the same value on both, while $\MUX$ and is not. Hence one of them is an error, and the algorithm returns it.
    
        \item[Branch at step 14] corresponds to Case 2 of Theorem~\ref{thm:improved-MUX-lb}.
        The algorithm chose the topological minimal gate $\alpha$ reading $x_i$ and found that $(\neg)\alpha$ is the output gate. After setting $x_i$ to a value that fixes $\alpha$, the restricted circuit becomes constant, while the restricted multiplexer is not. The algorithm considers inputs that differ on a data bit still relevant to the multiplexer and returns the one on which $\cC$ disagrees with $\MUX$.
    
        \item[Branch at step 19] corresponds to Case 3 of Theorem~\ref{thm:improved-MUX-lb}.
        Suppose that, after all variables in $D$ have been processed, the algorithm finds that $\cC$ does not read $a_{n-m+1}$. The circuit is independent of $a_{n-m+1}$, but the restricted multiplexer still depends on this bit, since it decides whether the selected data input lies in the lower or upper half. The addresses $\Call{bin}{1}$ and $\Call{bin}{2^{m-1} + 1}$ differ only in $a_{n-m+1}$, so $\cC$ outputs the same value on them. The algorithm sets $\vec{d}_1 = \vec{d}_{2^{m-1} + 1} + 1$, forcing $\MUX$ to evaluate differently on these two addresses. Thus one of the two must be a valid error.
    \end{description}
    
    Otherwise, the algorithm follows the recursive reduction from the \textbf{Inductive Hypothesis} of Theorem~\ref{thm:improved-MUX-lb}. Since the preceding branches were not entered, Assertions 1--3 from the lower bound proof hold. Thus each substitution for $x_i\in D$ eliminates at least two gates, and substituting $a_{n-m+1} \gets 0$ eliminates at least one additional gate. Therefore the circuit $\cC'$ passed to the recursive call satisfies
    \[
        \mu(\cC')
        < \bigl(2\cdot 2^m + m - 2\bigr) - (2^m + 1)
        = 2\cdot 2^{m - 1} + (m - 1) - 2.
    \]
    Moreover, after setting $a_{n-m+1} \gets 0$, the remaining function is the restricted multiplexer on the remaining $m-1$ live address bits and $2^{m-1}$ live data bits. Hence, the recursive call is valid. By the induction hypothesis, it returns an input on which the simplified circuit errs. Since the error only differs on active bits, the same input is also an error for the current circuit $\cC$.
    Thus, Algorithm~\ref{algo:mux-refuter-alt} correctly returns a valid witness $(\vec{w}^A,\vec{w}^X)$ such that
    $ \cC(\vec{w}^A,\vec{w}^X) \neq \vec{w}^X_{\Call{int}{\vec{w}^A}}$ 
    for all $m \in [n]$ by the principle of mathematical induction as desired.
    
    Finally, each recursive level performs only polynomial-time operations: checking whether a variable is read, finding a topological minimal gate, evaluating the circuit on a constant number of inputs, substituting, and simplifying. The recursion depth is at most $n=\log N$, so the total running time is polynomial in $N$. 
\end{proof}

\subsection{An Affine Refuter for Dimension $d$}
\label{sec:affine-refuter}

The $\mathbb{B}_2$ $3n-o(n)$ lower bound of \cite{DemenkovK11} for affine dispersers is constructive. For any $d : \N \to \N$, 

\begin{theorem}\label{thm:affine-
disperser-refutable}
There is an affine refuter for dimension $d$ against $\mathbb{B}_2$ $3n - 4d(n)$ size circuits.
\end{theorem}

This demonstrates that even gate elimination arguments which use (1) complex measures other than circuit size and (2) substitutions other than simple constant substitutions are constructive. To show this, we prove a stronger property: given any affine space of high enough dimension, we can find a subset of that space on which the circuit is constant. Recall from \cite{DemenkovK11} that $\mu(\cC) = \sigma(\cC) + \eta(\cC)$. We will abuse notation and write $d$ rather than $d(n)$.

\begin{lemma}\label{lemma:affine-subset-refuter}
    Fix $d : \N \to \N$. Let $A \in \F_2^{n \times d'}$ and $\vec{a} \in \F_2^{d'}$ define an affine space $S = \{A\vec{y} + \vec{a}\}$ of dimension $d' > d$. For any circuit $\cC$ such that $\mu(\cC) < 4(d' - d)$, Algorithm \ref{alg:affine-refuter} outputs $\{Wx + \vec{w}\} \subseteq S$ of dimension at least $d$ on which $\cC$ is constant in polynomial times.
\end{lemma}

To obtain an affine refuter, simply run Algorithm \ref{alg:affine-refuter} with $A = \mathbf{I}_n$ and $\vec{a} = \vec{0}$. Algorithm \ref{alg:affine-refuter} outputs an affine subspace of $\F_2^n$ of degree at least $d$ on which $\cC$ is constant. The algorithm mirrors the proof of the following theorem from \cite{DemenkovK11}. 

\begin{theorem}[from \cite{DemenkovK11}]\label{thm:affine-lb}
    Let $\func{f} : \{0,1\}^n \to \{0,1\}$ be an affine disperser of dimension $d$ and $S$ be a affine space of dimension $d' > d$. If $\cC$ computes $\func{f}$ on $S$ then $\mu(\cC) \geq 4(d' - d)$.
\end{theorem}

We begin by labeling proof of this theorem the proof as in the proof of Lemma \ref{lem:MUX-CktLB}.

\begin{proof}[Proof of Theorem \ref{thm:affine-lb} from \cite{DemenkovK11}]
    Let $\cC$ be a circuit computing $\func{f}$ on an affine space $S = \{A\vec{y} + \vec{a}\}$ of dimension at least $d$. If the dimension is exactly $d$ then the lower bound trivially holds. Otherwise we proceed by induction on the dimension of $S$. Fix $d' \geq d + 1$. 
    
    \begin{description}
        \item[Case 1: Computes $(\neg)\XOR$.] Suppose for the sake of contradiction, $\cC$ only contains fanout $1$ $\oplus$-type gates. Then $\cC$ computes $\bigoplus_{i \in I} x_i \oplus b$ for some $I \subset [n]$ and $c \in \{0,1\}$. However then $\cC$ outputs constant $0$ on $S \cap \{\vec{x} \in \F_2^n \mid \bigoplus_{i \in I} x_i = b\}$, an affine subspace whose dimension is at least $d' - 1 \geq d$
        \item[Assertion 1.] $\cC$ contains a gate which is not a fanout $1$ $\oplus$-type gate.
        \item[Work.] Let $\alpha$ be the topologically minimal gate which is not an $\oplus$-type gate with fanout $1$.
        \item[Cases over $\alpha$.] In each case, we find a substitution that reduces $\mu$ by at least $4$ and the resulting circuit computes $\func{f}$ on an affine space of dimension at least $d'-1$. Note if it computes a circuit on computes $\func{f}$ on an affine space of dimension at least $d'-1$, it computes $\func{f}$ on all subsets of that space with dimension exactly $d' - 1$.
            \begin{description}
                \item[Case 2: $\alpha$ is $\oplus$-type with fanout at least 2.] $\cC$ computes compute $\bigoplus_{i \in I} x_i \oplus b$ for some $I \subseteq[n]$ and $b \in \{0,1\}$. Substituting $x_j \gets \bigoplus_{i \in I \setminus \{j\}} x_i \oplus b$ for some $j \in I$ fixes $\alpha$ to output $0$ and reduces $\eta$ by $1$. Replacing $\alpha$ by $0$, eliminates at least the gates reading $\alpha$ reducing $\sigma$ by at least 3. Let $\cC'$ be the resulting circuit.
                    \begin{description}
                        \item[Case 2.1: Reduce Measure.] The circuit $\cC'$ satisfies $\mu(\cC') \leq \mu(\cC) - 4$ and $\cC'$ computes $\func{f}$ on $S \cap \{\vec{x} \in F_2^n \mid \bigoplus_{i \in I} x_i = b\}$, an affine space of dimension at least $d' - 1$.
                    \end{description}
                \item[Case 3: $\alpha$ is a $\land$-type gate fed by an $\oplus$-type gate $\beta$.] The gate $\beta$ computes $\bigoplus_{i \in I} x_i \oplus c$ for some $I$ and $c$. Let $b$ be the value that fixes $\alpha$. Substituting $x_j \gets \bigoplus_{i \in I \setminus \{j\}} x_i \oplus (b \oplus c)$ fixes $\beta$ to output $b$ which fixes $\alpha$. Let $\cC'$ be the resulting circuit.
                    \begin{description}
                        \item[Case 3.1: $\alpha$ is the output.] Suppose for the sake of contradiction $\alpha$ is the output. The function $\func{f}$ would be constant on $S \cap \{\vec{x} \in \F_2^n \mid \bigoplus_{i \in I} x_i = b \oplus c\}$ which is an affine space of dimension at least $d$. 
                        \item[Assertion 3.1.] $\alpha$ is not the output.
                        \item[Case 3.2: Reduce Measure] The substitution eliminates at least $3$ gates. $\mu(\cC')$ for the resulting circuit $\cC'$ is at least 4 less than $\mu(\cC)$, and $\cC'$ agrees with $\func{f}$ on the affine space $S \cap \{\vec{x} \in \F_2^n \mid \bigoplus_{i \in I} x_i = b \oplus c\}$.
                    \end{description}
                \item[Case 4: $\alpha$ is fed by two variables $x_p$ and $x_q$] Assume the fanout of $x_p$ is at least $x_q$. Let $b$ be the value that fixes $\alpha$. Substituting $x_p \gets \beta$ and simplifying will produce a circuit $\cC'$ where $\mu(\cC') \leq \mu(\cC) - 4$. We elide the complete case analysis here, see \cite{DemenkovK11} for more details. Each subcase (which depends on the fanout of $x_p$, the fanout of $\alpha$, and which specific gates $x_p$ feeds) reduces the complexity by at least four unless $\alpha$ is the output of the circuit. 
                \begin{description}
                    \item[Case 4.1 $\alpha$ is the output.] Suppose for the sake of contradiction $\alpha$ is the output. Then $\func{f}$ would be constant on $S \cap \{\vec{x} \in \F_2^n \mid x_p = c\}$
                    \item[Assertion 4.1] $\alpha$ is the not the output of the circuit.
                    \item[Case 4.2: Reduce Measure] Lastly, $\cC'$ computes $\func{f}$ on $S \cap \{\vec{x} \in \F_2^n \mid x_p = c\}$, an affine space of dimension at least $d' - 1$ and $\mu(\cC') \leq \mu(\cC) - 4$.
                \end{description}
            \end{description}
        \item[Case 5: Invoke Inductive Hypothesis] Each case resolves by reducing $\mu$ by at least $4$ and the simplified $\cC'$ computes $\func{f}$ on an affine subspace of dimension $d' - 1 \geq d$.  By induction, $\mu(\cC) \geq \mu(\cC') + 4 \geq 4((d'-1)-d) = 4(d' - d).$ \qedhere
    \end{description}
\end{proof}

\begin{algorithm}[t]
  \caption{Searches for an affine subspace on which $\cC$ is constant}
  \label{alg:affine-refuter}
  \begin{algorithmic}[1]
    \Require $\cC$ is a normalized circuit, $A \in \F_2^{n \times d'}$, $\vec{a} \in \F_2^n$, such that $A$ is full column rank, $d' > d$, and $\mu(\cC) < 4(d' - d)$ or $\mu(\cC) = 1$.
    \Ensure $W, \vec{w}$ such that $\cC$ is constant on $\{W\vec{x} + \vec{w}\} \subseteq \{A\vec{y} + \vec{a}\}$ and $\dim W \geq d$
    \Procedure{FindConstantSubspace}{$\cC$, $A$, $\vec{a}$}
    \If{$\cC$ is constant}
        \Comment{Case 3.1 and 4.1}
        \State \Return $A$, $\vec{a}$
    \EndIf
    \If{$\cC$ is a circuit of $\oplus$-type gates with fan-out $1$}
        \Comment{Case 1}
        \State $I, c \gets$ \Call{ParitySupport}{$\cC$} \Comment{$\cC$ computes $\bigoplus_{i \in I} x_i \oplus c$}
        \State $R, \vec{r} \gets$ \Call{ConstraintToAffine}{I, c} \Comment{$\{Rz + \vec{r}\} = \{x \in \F_2^n \mid \bigoplus_{i \in I} x_i = c\}$}
        \State \Return \Call{AffineIntersect}{$A, \vec{a}, R, \vec{r}$} \Comment{$C$ is constantly $0$ on this intersection}
    \EndIf
    
    \State $\alpha \gets$ topologically minimal non-fanout $1$ $\oplus$-type gate of $\cC$
    \If{$\alpha$ is an $\oplus$-type gate}
    \Comment{Case 2}
            \State $\calP \gets$ the induced sub-circuit of $\cC$ rooted at $\alpha$
            \State $I, c \gets$ \Call{ParitySupport}{$\calP$} \Comment{$\calP$ computes $\bigoplus_{i \in I} x_i \oplus c$}
            \State $\cS, j \gets$ \Call{FindSubstitution}{$\calP, 0$} \Comment{$\cS = \bigoplus_{i \in I \setminus \{j\}} x_i \oplus c$}
            \State $\cC \gets \cC$ where $x_j$ is substituted with $\cS$, $\alpha$ is replaced by $0$, and then simplified
            \State $R, \vec{r} \gets$ \Call{ConstraintToAffine}{$I, c$}
        \ElsIf{a $\oplus$-type gate $\beta$ feeds $\alpha$}
            \Comment{Case 3}
            \State $\cQ \gets$ induced sub-circuit of $\cC$ rooted at $\beta$
            \State $c' \gets $ the constant value that fixes $\alpha$
            \State $I, c \gets$ \Call{ParitySupport}{$\cQ$} \Comment{$\cQ$ computes $\bigoplus_{i \in I} x_i \oplus c$}
            \State $\cS, j \gets$ \Call{FindSubstitution}{$\cQ, c'$} \Comment{$\cS = \bigoplus_{i \in I \setminus \{j\}} x_i \oplus (c + c')$}
            \State $\cC \gets \cC$ where $x_j$ is substituted with $\cS$, $\beta$ is replaced by $c'$, and then simplified
            \State $R, \vec{r} \gets$ \Call{ConstraintToAffine}{$I, c \oplus c'$}
        \Else
            \Comment{Case 4: $\alpha$ is fed by two variables}
            \State $x_j \gets $ the variable feeding $\alpha$ with maximal fanout
            \State $c' \gets$ the constant value that fixes $\alpha$
            \State $\cC \gets \cC$ where $x_j$ is substituted with $c'$ and then simplified 
            \State $R, \vec{r} \gets$ \Call{ConstraintToAffine}{$\{j\}, c'$}
        \EndIf
        \State $A, \vec{a} \gets$ \Call{AffineIntersect}{$A, \vec{a}, R, \vec{r}$}
        \State \Return \Call{FindConstantSubspace}{$\cC, A, \vec{a}$}
    \EndProcedure
  \end{algorithmic}
\end{algorithm}

We now prove that Algorithm \ref{alg:affine-refuter} correctly outputs an affine subspace of sufficient size on which its input circuit is constant. 

\begin{proof}[Proof of Lemma \ref{lemma:affine-subset-refuter}]
    Fix $d : \N \to \N$. Let $A \in \F_2^{d'}$ and $\vec{a} \in \F_2^{d'}$ define an affine space $S$ of dimension $d' > d$. Let $\cC$ be a circuit where $\mu(\cC) < 4(d' - d)$. We will show that Algorithm \ref{alg:affine-refuter}, \textsc{FindConstantSubspace}($\cC, A, \vec{a})$, outputs a space $S'$ of dimension $\geq d$ such that $S' \subseteq S$ in polynomial time with respect to $n$ and $\cC$ is constant on $S'$. \textsc{FindConstantSubspace} uses four polynomial time subroutines defined in Appendix \ref{sec:affine-subroutines}. We summarize these in Table \ref{tab:affine-subroutines}.

    \begin{table}[h]
    \caption{The four polynomial subroutines used by \textsc{FindConstantSubspace}}\label{tab:affine-subroutines}
\begin{tabular}{|l|l|}
\hline
Procedure                                    & Description                                                                                                                                 \\
\hline
\textsc{ParitySupport}      & Given $\cC$ with only $\oplus$-type gates of fanout 1 \\ & Returns $I, c$ such that $\cC$ computes $\bigoplus_{i \in I} x_i \oplus c$           \\
\hline
\textsc{AffineIntersect}    & Given two affine spaces $S$ and $S'$ of dimension $d'$ and $n-1$ \\ &Returns a description of their dimension $\geq d'-1$ intersection.                                                                       \\
\hline
\textsc{ConstraintToAffine} & Given $(I, c)$ where $I \subseteq [n]$ and $c \in \{0,1\}$ \\ & Returns the description of the affine space  $\{x_i \in \F_2^n \mid \bigoplus_{i \in I} x_i = c\}$              \\
\hline
\textsc{FindSubstitution}   & Given $(\cC$, b), a circuit with only $\oplus$-type gates of fanout 1 \\ & Returns the substitution $x_j \gets \cS$ that makes $\cC$ output $b$ \\
\hline
\end{tabular}
\end{table}

It is clear that if the algorithm terminates then the output is a subset of $S$, as in each call we only take the \emph{intersection} of affine spaces, which are subsets of our starting space. We argue that each iteration, if it does not return in branches 2 or 5, reduces $\mu(\cC)$ by at least 4 or makes the circuit constant. This follows as the algorithm mirrors the structure of the proof of theorem \ref{thm:affine-lb}; subsequent branches correspond to the cases over $\alpha$.

\begin{description}
    \item[Branch at step 13] corresponds to Case 2. The substitution from the proof removes at least three gates and reduces $\eta(\cC)$ by at least $1$.
    \item[Branch at step 21] corresponds to Case $3$. The substitution from the proof either reduces $\mu$ by four or makes the circuit constant.
    \item[Branch at step 30] corresponds to Case $4$. The substitution from the proof either reduces $\mu$ by four or makes the circuit constant.
\end{description}

If, $\cC$ becomes constant as in Cases 3.1 or 4.1, the next recursive call enters branch 2 and returns the intersection on which $\cC$ is constant described in the proof. Otherwise, $\mu(\cC)$ is reduced by $4$. Notice that we can make at most $(d' - d)$ iterations before $\cC$ must become constant. Hence the algorithm terminates and returns an affine space on which $\cC$ is constant. During each iteration, $\textsc{ConstraintToAffine}$, reduces the dimension by at most $1$, hence the returned space has dimension at least $d$.

Finally, as each subroutine runs in polynomial time and we make at most $d' - d$ iterations, \textsc{FindConstantSubspace} runs in polynomial time.
\end{proof}


\clearpage



\bibliography{references}

@article{Redkin1973,
    title={Proof of minimality of circuits consisting of functional elements},
    author={Red’kin, NP},
    journal={Systems Theory Research: Problemy Kibernetiki},
    pages={85--103},
    year={1973},
    publisher={Springer}
}

@article{Schnorr74,
    author    = {Claus{-}Peter Schnorr},
    title     = {Zwei lineare untere Schranken f{\"{u}}r die Komplexit{\"{a}}t
               Boolescher Funktionen},
    journal   = {Computing},
    volume    = {13},
    number    = {2},
    pages     = {155--171},
    year      = {1974},
    url       = {https://doi.org/10.1007/BF02246615},
    doi       = {10.1007/BF02246615},
    bibsource = {dblp computer science bibliography, https://dblp.org}
}

@inproceedings{Paul1975,
  title={A 2.5 n-lower bound on the combinational complexity of Boolean functions},
  author={Paul, Wolfgang J},
  booktitle={Proceedings of the seventh annual ACM symposium on Theory of computing},
  pages={27--36},
  year={1975}
}

@article{KleinP1980,
  title={Asymptotically optimal circuit for a storage access function},
  author={Klein and Paterson},
  journal={IEEE Transactions on Computers},
  volume={100},
  number={8},
  pages={737--738},
  year={1980},
  publisher={IEEE}
}

@article{BlumK95,
  author     = {Blum, Manuel and Kannan, Sampath},
  title      = {Designing programs that check their work},
  year       = {1995},
  issue_date = {Jan. 1995},
  publisher  = {Association for Computing Machinery},
  address    = {New York, NY, USA},
  volume     = {42},
  number     = {1},
  issn       = {0004-5411},
  url        = {https://doi.org/10.1145/200836.200880},
  doi        = {10.1145/200836.200880},
  journal    = {Journal of the ACM (JACM)},
  month      = jan,
  pages      = {269–291},
  numpages   = {23}
}

@inproceedings{DemenkovK11,
    author    = {Evgeny Demenkov and
               Alexander S. Kulikov},
    editor    = {Filip Murlak and
               Piotr Sankowski},
    title     = {An Elementary Proof of a 3n - o(n) Lower Bound on the Circuit Complexity
               of Affine Dispersers},
    booktitle = {Mathematical Foundations of Computer Science 2011 - 36th International
               Symposium, {MFCS} 2011, Warsaw, Poland, August 22-26, 2011. Proceedings},
    series    = {Lecture Notes in Computer Science},
    volume    = {6907},
    pages     = {256--265},
    publisher = {Springer},
    year      = {2011},
    url       = {https://doi.org/10.1007/978-3-642-22993-0\_25},
    doi       = {10.1007/978-3-642-22993-0\_25},
    bibsource = {dblp computer science bibliography, https://dblp.org}
}

@INPROCEEDINGS{FindGHK2016,
    author={Find, Magnus Gausdal and Golovnev, Alexander and Hirsch, Edward A. and Kulikov, Alexander S.},
    booktitle={2016 IEEE 57th Annual Symposium on Foundations of Computer Science (FOCS)}, 
    title={A Better-Than-3n Lower Bound for the Circuit Complexity of an Explicit Function}, 
    year={2016},
    volume={},
    number={},
    pages={89-98},
    doi={10.1109/FOCS.2016.19}
}

@article{LozhkinK2021,
  title={The Complexity of the Standard Multiplexer Function in a Class of Switching Circuits},
  author={Lozhkin, SA and Khzmalyan, DE},
  journal={Computational Mathematics and Modeling},
  volume={32},
  number={4},
  pages={478--489},
  year={2021},
  publisher={Springer}
}

@inproceedings{Li022,
    author    = {Jiatu Li and
               Tianqi Yang},
    editor    = {Stefano Leonardi and
               Anupam Gupta},
    title     = {3.1\emph{n} - \emph{o}(\emph{n}) circuit lower bounds for explicit
               functions},
    booktitle = {{STOC} '22: 54th Annual {ACM} {SIGACT} Symposium on Theory of Computing,
               Rome, Italy, June 20 - 24, 2022},
    pages     = {1180--1193},
    publisher = {{ACM}},
    year      = {2022},
    url       = {https://doi.org/10.1145/3519935.3519976},
    doi       = {10.1145/3519935.3519976},
    bibsource = {dblp computer science bibliography, https://dblp.org}
}

@article{ChenJSW24,
  author    = {Lijie Chen and
               Ce Jin and
               Rahul Santhanam and
               Ryan Williams},
  title     = {Constructive Separations and Their Consequences},
  journal   = {TheoretiCS},
  volume    = {3},
  year      = {2024},
  url       = {https://doi.org/10.46298/theoretics.24.3},
  doi       = {10.46298/THEORETICS.24.3},
  bibsource = {dblp computer science bibliography, https://dblp.org}
}

@inproceedings{GrosserC25,
  author       = {Stefan Grosser and
                  Marco Carmosino},
  editor       = {Michal Kouck{\'{y}} and
                  Nikhil Bansal},
  title        = {Student-Teacher Constructive Separations and (Un)Provability in Bounded
                  Arithmetic: Witnessing the Gap},
  booktitle    = {Proceedings of the 57th Annual {ACM} Symposium on Theory of Computing,
                  {STOC} 2025, Prague, Czechia, June 23-27, 2025},
  pages        = {1341--1347},
  publisher    = {{ACM}},
  year         = {2025},
  url          = {https://doi.org/10.1145/3717823.3718216},
  doi          = {10.1145/3717823.3718216},
  bibsource    = {dblp computer science bibliography, https://dblp.org}
}

@inproceedings{CarmosinoDJ2025,
  author    = {Carmosino, Marco and Dang, Ngu and Jackman, Tim},
  title     = {{Simple Circuit Extensions for XOR in PTIME}},
  booktitle = {43rd International Symposium on Theoretical Aspects of Computer Science (STACS 2026)},
  pages     = {23:1--23:20},
  series    = {Leibniz International Proceedings in Informatics (LIPIcs)},
  isbn      = {978-3-95977-412-3},
  issn      = {1868-8969},
  year      = {2026},
  volume    = {364},
  editor    = {Mahajan, Meena and Manea, Florin and McIver, Annabelle and Thang, Nguyen Kim},
  publisher = {Schloss Dagstuhl -- Leibniz-Zentrum f{\"u}r Informatik},
  address   = {Dagstuhl, Germany},
  url       = {https://drops.dagstuhl.de/entities/document/10.4230/LIPIcs.STACS.2026.23},
  urn       = {urn:nbn:de:0030-drops-255127},
  doi       = {10.4230/LIPIcs.STACS.2026.23}
}

@article{DBLP:journals/tcs/Blum84,
  author       = {Norbert Blum},
  title        = {A Boolean Function Requiring 3n Network Size},
  journal      = {Theor. Comput. Sci.},
  volume       = {28},
  pages        = {337--345},
  year         = {1984},
  url          = {https://doi.org/10.1016/0304-3975(83)90029-4},
  doi          = {10.1016/0304-3975(83)90029-4},
  bibsource    = {dblp computer science bibliography, https://dblp.org}
}

@article{DBLP:journals/siamcomp/Ilango24,
  author       = {Rahul Ilango},
  title        = {Constant Depth Formula and Partial Function Versions of {MCSP} Are
                  Hard},
  journal      = {{SIAM} J. Comput.},
  volume       = {53},
  number       = {6},
  pages        = {S20--317},
  year         = {2024},
  url          = {https://doi.org/10.1137/20m1383562},
  doi          = {10.1137/20M1383562},
  bibsource    = {dblp computer science bibliography, https://dblp.org}
}

@article{DBLP:journals/jcss/GolovnevHKK18,
  author       = {Alexander Golovnev and
                  Edward A. Hirsch and
                  Alexander Knop and
                  Alexander S. Kulikov},
  title        = {On the limits of gate elimination},
  journal      = {J. Comput. Syst. Sci.},
  volume       = {96},
  pages        = {107--119},
  year         = {2018},
  url          = {https://doi.org/10.1016/j.jcss.2018.04.005},
  doi          = {10.1016/J.JCSS.2018.04.005},
  bibsource    = {dblp computer science bibliography, https://dblp.org}
}

@article{DBLP:journals/apal/MullerP20,
  author       = {Moritz M{\"{u}}ller and
                  J{\'{a}}n Pich},
  title        = {Feasibly constructive proofs of succinct weak circuit lower bounds},
  journal      = {Ann. Pure Appl. Log.},
  volume       = {171},
  number       = {2},
  year         = {2020},
  url          = {https://doi.org/10.1016/j.apal.2019.102735},
  doi          = {10.1016/J.APAL.2019.102735},
  bibsource    = {dblp computer science bibliography, https://dblp.org}
}

@book{Jukna,
  author       = {Stasys Jukna},
  title        = {Boolean Function Complexity - Advances and Frontiers},
  series       = {Algorithms and combinatorics},
  volume       = {27},
  publisher    = {Springer},
  year         = {2012},
  url          = {https://doi.org/10.1007/978-3-642-24508-4},
  doi          = {10.1007/978-3-642-24508-4},
  isbn         = {978-3-642-24507-7},
  bibsource    = {dblp computer science bibliography, https://dblp.org}
}

\clearpage
\appendix
\clearpage
\section{An Optimal $\XOR$ Checker}\label{sec:xor-checker}

Refuters find errors for computational models which fall \emph{below} a threshold established from a lower bound. It is natural to ask whether one can still find errors the model \emph{meets} the threshold but still fails to compute $\func{f}$. We call such an algorithm a \emph{checker}. We show that checkers exist for $\XOR$.
It will be useful to keep in mind the following facts about $\XOR$.
\subparagraph*{Properties of $(\neg)\XOR$.} $(\neg)\XOR_n$ is fully downward self-reducible: fixing any subset of inputs leaves $(\neg)\XOR$ on the remaining variables. Second, $(\lnot)\XOR_n$ strongly depends on every input: flipping any single bit always flips the output. Taken together, every partial restriction of $(\lnot)\XOR_n$ remains non-degenerate.

\begin{theorem}
  \label{thm:refined-xor-refuter}
    Let $C$ be a DeMorgan circuit on $n$ inputs with size $s = 3(n - 1)$ that \textbf{does not} compute $(\neg)\XOR_n$. Then, an input $x \in \{0, 1\}^n$ such that $C(x) \neq (\neg)\XOR_n(x)$ can be found in polynomial time with respect to $n$.
\end{theorem}

While refuters follow straightforwardly from proofs via gate elimination, Theorem \ref{thm:refined-xor-refuter} does not. It leverages a strong structural characterization of the circuits which compute $\XOR$. Section \ref{sec:xor-refuter} establishes a refuter subroutine for $\XOR$ which serves as the basis for the checker. In Section \ref{sec:xor-detector}, we show that we can \emph{detect} optimal $(\neg)\XOR$ circuits due to to their rigid structure. Finally, in Section \ref{sec:xor-checker}, we leverage the refuter and the detectability of $\XOR$ to obtain our checker.

\subsection{A Refuter for $\XOR$}\label{sec:xor-refuter}

We begin with an $\XOR$ refuter (Algorithm \ref{alg:xor-refuter}), which our checker can both use as a subroutine, and whose procedure captures many ways in which an appropriately sized circuit can still ``trivially'' fail to compute $\XOR$.

\begin{theorem}
  \label{thm:xor-refuter}
    Let $C$ be a DeMorgan circuit on $n \geq 2$ inputs with size $s < 3(n - 1)$. We can find an input $x \in \{0, 1\}^n$ such that $C(x) \neq \XOR_n(x)$ in polynomial time with respect to $n$.
\end{theorem}

\begin{proof}

Let $\cC$ be a DeMorgan circuit on $n$ inputs with size $s < 3(n - 1)$ and let $f$ be the function it computes. We will show Algorithm \ref{alg:xor-refuter}, when run on $\cC$, $I = [n], b = 0 (1 \text{ for } (\neg)\XOR), \vec{a} = \vec{0}$ outputs an input on which $\cC$ errs in polynomial time. Termination and correctness of the algorithm follow from the fact that it follows the structure of Schnorr's $3(n-1)$ lower bound, modulo the proof that $\cC$ actually errs on one of the inputs at each \textbf{compute} step. We recall Schnorr's argument, structuring it so that the underlying algorithm is clear.

\begin{proof}[$(\neg)\XOR$ Lower Bound from \cite{Schnorr74}] As a proof via gate elimination, the argument proceeds via induction. The base case is trivial as $(\neg)\XOR_1(x) \equiv (\neg)x$ and $3(1 - 1) = 0$. Let $C$ be a normalized circuit with $n \geq 2$ inputs. 
  \begin{description}
  \item[Case 1: Obvious Degeneracy.] Suppose $C$ does not read some $x_i$.  Then $C$ cannot depend on $x_i$, and therefore does not compute $\XOR_n$, which depends on all $n$ variables.

  \item[Assertion 1.] $C$ reads every input variable at least once.

  \item[Work: Capture the First Gate.] Let $\alpha$ be the topologically minimal binary gate in $C$.  Because $C$ is in normal form, $\alpha$ reads $(\neg)x_i$ and $(\neg)x_j$ for some \emph{distinct} $i,j \in [n]$.

  \item[Case 2: Insufficient Fanout.] Without loss of generality, fix $i$ and suppose $\alpha$ is the \emph{only} costly gate reading from $x_i$, so the fanout of $x_i$ is 1.  If we substitute $x_j$ to fix $\alpha$ and simplify, $x_i$ becomes disconnected from $C$.  The resulting circuit cannot compute $(\neg)\XOR_{n-1}$, which still depends on $x_i$. Therefore $C$ cannot compute $\XOR_n$.

  \item[Assertion 2.] The fanout of $x_i$ is strictly greater than 1.

  \item[Case 3: Short-Circuit at $\alpha$.] Suppose $\alpha$ is the output gate.  Without loss of generality, fix $i$ and substitue $x_j$ to eliminate $\alpha$ via a fixing rule. This trivializes the entire circuit $C$ by simplifying it to a constant. Therefore, $C$ cannot compute $\XOR_n$ as a single bit restriction of $\XOR_n$ is $(\neg)\XOR_{n-1}$.  This is the same contradiction obtained above.

  \item[Assertion 3.] Gate $\alpha$ is not the output.

  \item[Work: Capture a Second Gate.] Let $\beta$ be a topologically minimal binary gate distinct from $\alpha$ that also reads $x_i$ in $C$.  $\beta$ must exist because the fanout of $x_i$ is strictly greater than 1.

  \item[Case 4: Short-Circuit at $\beta$.] Suppose $\beta$ is the output gate and argue exactly as in Case 3 using $x_i$ to fix $\beta$.

  \item[Assertion 4.] Gate $\beta$ is not the output.

  \item[Work: Capture a Third Gate.]  Let $\gamma$ be the topologically minimal binary gate that reads $\beta$. Fix $x_i$ to fix $\beta$ and simplify. This eliminates $\alpha, \beta$ and $\gamma$, reducing the circuit size by $3$.

  \item[Invoke Inductive Hypothesis.] $\cC$ now computes $(\neg)\XOR_{n-1}$ and its size has reduced by at least three. Applying our inductive hypothesis, the original circuit had size at least $3 + \sigma((\neg)\XOR_{n-1}) = 3 + 3((n-1)-1) = 3n - 3 = 3(n - 1)$. \qedhere
  \end{description}
\end{proof}

We now argue that the algorithm outputs an error. Each recursive call of the algorithm corresponds to the inductive step for $|I| = k \leq n$. In the base case, when $k = 2$ and $\sigma(\cC) < 3$, the algorithm in steps 2 - 6 brute forces over all four possible assignments and returns one on which it errs by construction. We proceed via induction; fix $k > 2$ to be the size of $I$ in our input and presume \textsc{XOR-Refuter} correctly outputs on all inputs where $|I| = k - 1$. The algorithm does not enter the branch at step $2$ as $|I| = k > 2$. Subsequent branches correspond to the Cases identified in Schnorr's proof. If we do not enter these branches then the corresponding Assertions of Schnorr's proof hold for $\cC$. Hence any gates and variables bound in the interleaving steps exist. 

\begin{description}
    \item[Branch at step 7] corresponds to Case 1 of Schnorr's proof. $\cC$ does not depend on $x_i$ but $\XOR_I(\cdot) \oplus c$ does. In step 9, $\cC$ evaluates the same on both inputs but $\XOR_I(\cdot) \oplus c$ does not. The algorithm returns whichever is incorrect.

 \item[Branch at step 13] corresponds to Case 2 and 3. Setting $x_q$ to fix $\alpha$ in $\cC$ either disconnecting $x_p$ or makes the circuit constant. $\cC$ evaluates the same on both inputs unlike $\XOR_I(\cdot) \oplus c$. 

 \item[Branch at step 19] corresponds to Case 4. Setting $x_p$ as described leaves $\cC$ to be a constant, but $\XOR_{I}(\cdot) \oplus c$ does not become constant. In this and the previous branches, the only indices of $\vec{a}$ that are changed are in $I$, and hence the condition that the error agrees with the original $\vec{a}_i$ on all non-$I$ inputs is met.
\end{description}

Otherwise, the algorithm in steps 25 - 28 follows the final Work and Inductive Hypothesis step of the proof. As Assertions 1-4 hold for $\cC$, setting $x_p$ to fix $\beta$ eliminates at least three gates, and the resulting circuit is now too small to compute the appropriate parity of the $k-1$ bits indexed by ${I \setminus \{p\}}$. Therefore the recursive call returns an input on which the simplified circuit errs that is consistent with $\vec{a}_p$ (and all $\vec{a}_i$ where $i \not\in I)$. Hence $\cC$ also errs on this input as desired and this error is consistent with $\vec{a}$ on all non-$I$ indices. Algorithm \ref{alg:xor-refuter} is therefore correct. 

Each iteration can clearly be computed in polynomial time, as $\sigma(\cC) < 3(n-1)$ and the normalization of $\cC$ ensure that evaluating, substituting, and simplifying can all be done in polynomial time.
\end{proof}

\begin{algorithm}[p]
  \caption{Refuter for circuits purportedly computing $(\neg)\XOR_n$ with fewer than $3(n-1)$ gates}
  \label{alg:xor-refuter}
  \begin{algorithmic}[1]
    \Require $\cC$ is a normalized circuit, $I \subseteq [n]$, $c \in \{0,1\}$, $\vec{a} \in \F_2^n$, such that $\sigma(\cC) < 3(|I|-1)$ and $\cC$ only reads variables in $I$
    \Ensure $\vec{w} \in \F_2^n$ such that $\cC(\vec{w}) \neq \XOR_{I}(\vec{w}) \oplus c$ and $\vec{w}_i = \vec{a}_i$ for all $i \not\in I$
    \Procedure{XOR-Refuter}{$\cC, I, c, \vec{a}$}
    \If{$|I| = 2$}
    \LComment{Base Case: only 4 assignments possible and $\cC$ must err on one}
        \State $i, j \leftarrow$ elements of $I$
        \State \textbf{compute} $\cC$ on $\vec{a}, \vec{a} + \vec{e}_i, \vec{a} + \vec{e_j}, \vec{a} + \vec{e}_i + \vec{e}_j$
        \State \Return an input on which $\cC \neq \XOR_I$
    \EndIf
    \If{$\cC$ does not read $x_i$ for some $i \in I$}
    \LComment{Case 1: $\cC$ is degenerate with respect to $x_i$ but $\XOR_I$ is not}
        \State \textbf{compute} $\cC$ on $\vec{a}$ and $\vec{a} + \vec{e}_i$
        \State \Return whichever input $\cC$ does not agree with $\XOR_{I}(\cdot) \oplus c$ on
    \EndIf
    \State $\alpha \leftarrow$ topological minimal binary gate in $\cC$
    \State $p, q \leftarrow $ indices of $I$ such that $\alpha$ reads $(\neg)x_p$ and $(\neg)x_q$
    \If{\Call{Fanout}{$x_p$} = 1 or $(\neg)\alpha$ is the output}
        \State $\vec{a}_q \leftarrow $ constant whose substitution for $x_q$ fixes $\alpha$
        \LComment{Case 2 and 3: $\cC\hook_{x_q \gets \vec{a}_q}$ no longer depends on $x_p$ but $\XOR_{I \setminus \{q\}}$ does}
        \State \textbf{compute} $\cC$ on $\vec{a}$ and $\vec{a} + \vec{e}_p$
        \State \Return whichever input $\cC$ does not agree with $\XOR_{I}(\cdot) \oplus c$ on
    \EndIf
    \State $\beta \gets$ topological minimal binary gate reading $(\neg)x_p$ distinct from $\alpha$
    \If{$(\neg)\beta$ is the output gate of $\cC$}
        \State $\vec{a}_p \gets$ constant whose substitution for $x_p$ fixes $\beta$
        \LComment{Case 4: $\cC\hook_{x_p \gets \vec{a}_p}$ is constant but $\XOR_{I \setminus \{p\}}$ is not}
        \State \textbf{compute} $\cC$ on $\vec{a}$ and $\vec{a} + \vec{e}_p$
        \State \Return whichever input $\cC$ does not agree with $\XOR_{I}(\cdot) \oplus c$ on
    \EndIf
    \LComment{We can eliminate at least three gates and recurse}
    \State $\gamma \gets $ topologically minimal binary gate reading $(\neg)\beta$
    \State $\vec{a}_p \gets$ constant whose substitution for $x_p$ fixes $\beta$
    \State $\cC \gets \cC$ with $\vec{x}_p$ substituted for $x_p$ and then simplified
    \State \Return \Call{XOR-Refuter}{$\cC, I \setminus \{p\}, b \oplus \vec{a}_p, \vec{a}$}
    \EndProcedure
  \end{algorithmic}
\end{algorithm}

\clearpage
\subsection{Optimal $\XOR$ Circuits Are Detectable}
\label{sec:xor-detector}
Algorithm \ref{alg:xor-checker}, the $\XOR$ checker, initially proceeds in the same way as the refuter. However, when it finds a substitution that eliminates at least three gates, it is possible that the restricted circuit \emph{does} compute the appropriate $\XOR$ of the remaining inputs, and $\cC$ when the input is restricted with the opposite bit. To circumvent  this, the checker will need to be able to detect whether the resulting circuit is now correct. We show this can be done in polynomial time.

\begin{lemma}
\label{lem:xor-detector}
    Let $\cC$ be a normalized DeMorgan circuit on $n$ variables of size $3(n-1)$. Determining whether $\cC$ computes $(1) \XOR_n$, $(2) \neg\XOR_n$, or $(3)$ neither can be computed in polynomial time with respect to $n$.
\end{lemma}

Computing the truth table in exponential in $n$. Improving over naive brute force is possible because the structure of optimal circuits computing $\XOR$ in the DeMorgan basis (when $\neg$ gates do not contribute to circuit size) has been characterized exactly \cite{CarmosinoDJ2025}.

\begin{theorem}[from \cite{CarmosinoDJ2025}]
    \label{thm:XOR-structure}
    When $\neg$ gates do not contribute to circuit size, optimal $(\neg)\XOR$ circuits in the DeMorgan basis are trees of $n-1$ $(\neg)\XOR_2$ widgets.  
\end{theorem}

We can therefore partition $\cC$ into $(\neg)\XOR_2$ widgets. If this is impossible than $\cC$ does not compute $(\neg)\XOR_n$. To prevent false positives, we need the following claim which is the converse of Theorem \ref{thm:XOR-structure}.

\begin{lemma} \label{lem:XOR-id-correctness}
    Let $C$ be a circuit of size $3(n-1)$ and let $f: \{0, 1\}^n \rightarrow \{0, 1\}$ be the Boolean function it computes. If $C$ can be partitioned into $(\neg)\XOR_2$ widgets then $f \equiv (\lnot)\XOR_n$. 
\end{lemma}

\begin{proof}
    This follows via strong induction and verifying that taking $(\neg)\XOR_2$ of two circuits computing $(\neg)\XOR$ on distinct subsets of variables computes $(\neg)\XOR$ of their union.
\end{proof}

The characterization and the lack of false positives together yields the lemma.

\begin{proof}{Proof of Lemma \ref{lem:xor-detector}}
    To determine whether $\cC$ computes $\XOR_n$ (or $\neg \XOR_n)$, partition the circuit into $n-1$ blocks, each of size $3$, as described in the proof of Theorem \ref{thm:XOR-structure}. As there are a finite number of normalized optimal $(\neg)\XOR_2$ blocks, we simply hardcode a list of them in our algorithm to ensure that each block is an $(\neg)\XOR_2$ widget. If any block does not computes $(\neg)\XOR_2$ then we \textbf{reject}. Lastly, we determine whether $\cC$ computes $\XOR_n$ or $\neg\XOR_n$ we can evaluate $\cC$ on the all zero input: if it evaluates to $0$ then $C$ computes $\XOR_n$ rather than $\neg\XOR_n$. Since $\cC$ is normalized each of the above steps runs in polynomial time with respect to $n$.
\end{proof}

\subsection{An $\XOR$ Checker}
\label{sec:improved-refuter}
We now show Theorem \ref{thm:refined-xor-refuter} by giving an efficient procedure (Algorithm \ref{alg:xor-checker}) that leverages the refuter (Algorithm \ref{alg:xor-refuter}) and the detectability of $\XOR$. 

\begin{algorithm}[t]
  \caption{Checker for circuits purportedly computing $(\neg)\XOR_n$}
  \label{alg:xor-checker}
  \begin{algorithmic}[1]
    \Require $\cC$ is a normalized circuit, $I \subseteq [n]$, $c \in \{0,1\}$, $\vec{a} \in \F_2^n$, such that $\sigma(\cC) = 3(|I|-1)$ and $\cC$ only reads variables in $I$ and $\cC$ does not compute $\XOR_I \oplus c$
    \Ensure $\vec{w} \in \F_2^n$ such that $\cC(\vec{w}) \neq \XOR_{I}(\vec{w}) \oplus c$ and $\vec{w}_i = \vec{a}_i$ for all $i \not\in I$
    \Procedure{XOR-Checker}{$\cC, I, c, \vec{a}$}
    \State Run lines 2 - 23 of \textsc{XOR-Refuter}
    \setcounter{ALG@line}{23}  
    \LComment{We can eliminate at least three gates and \emph{try} to recurse}
     \State $\gamma \gets $ topologically minimal binary gate reading $(\neg)\beta$
     \State $\vec{a}_p \gets$ constant whose substitution for $x_p$ fixes $\beta$
     \State $\cC' \gets \cC$ with $\vec{a}_p$ substituted for $x_p$ and then simplified
     \If{$|\cC'| < |\cC| - 3$}
        \LComment{$\cC\hook_{x_p \gets \vec{a}_p}$ is too small to compute $(\neg) \XOR_{I \setminus \{p\}}$}
        \State \Return \Call{XOR-Refuter}{$\cC', I \setminus \{p\}, c \oplus \vec{a}_p, \vec{a}$}
     \ElsIf{$\cC'$ does not computes $\XOR_{I \setminus \{p\}} (\cdot )\oplus (c \oplus \vec{a}_p)$}
        \State \Return \Call{XOR-Checker}{$\cC', I \setminus \{p\} \oplus \vec{a}_p, \vec{a}$}
     \EndIf
     \LComment{$\cC\hook_{x_p \gets \vec{a}_p}$ \emph{does} compute $\XOR_{I \setminus \{p\}}(\cdot) \oplus (c \oplus \vec{a}_p)$; must flip $\vec{a}_p$ to find error}
     \State $\vec{a}_p \gets 1 + \vec{a}_p$
    \If{$\alpha$ only feeds $\beta$ as in Figure \ref{fig:checker-corner-case-1}}
        \State $\tilde{\cC} \gets$ $\cC$ with $\alpha$ removed and $\beta$ replaced by an optimal circuit computing the same function as $\beta$
        \LComment{$\tilde{\cC}$ computes the same function as $\cC$ but $\sigma(\tilde{\cC}) < 3(|I|-1)$}
        \State \Return \Call{XOR-Refuter}{$\tilde{\cC}, I, b, \vec{a}$}
    \EndIf
    \LComment{$x_p \gets 1 + \vec{a}_p$ also eliminates at least three gates}
    
    \State $\cC \gets \cC$ with $x_p$ substituted by $\vec{a}_p$ and simplified
    \State \Return \Call{XOR-Checker}{$\cC$, $I \setminus \{p\}, c \oplus \vec{a}_p$, $\vec{a}$}

    \EndProcedure
  \end{algorithmic}
\end{algorithm}

\begin{proof}[Proof of Theorem \ref{thm:refined-xor-refuter}]
    Let $\cC$ be a DeMorgan circuit on $n$ inputs of size $3(n-1)$ which does not compute $\XOR_n$. We will argue that \textsc{XOR-Checker}, Algorithm \ref{alg:xor-checker}, outputs an input on which $\cC$ fails to compute $\XOR_n$ in polynomial time when run with $I = [n], \vec{a} = \vec{0}$. 

    \textsc{XOR-Checker} proceeds as the \textsc{XOR-Refuter} up to line 23. In the base case, $n = 2$, the algorithm brute forces over all four possible inputs and $\cC$ must err on at least one of them else it computes $\XOR_2$. We proceed via induction; fix $k > 2$ to be the size of $I$ and presume that \textsc{XOR-Checker} is correct on any inputs where $|I| = k-1$.
    
    If the algorithm branches before line 24, then $\cC$ must err as in the proof of correctness for \textsc{XOR-Refuter}. Otherwise, Assertions 1 - 4 of Schnorr's proof hold for $\cC$ and setting $x_p$ to fix $\beta$ will eliminate at least three gates. However, unlike in \textsc{XOR-Refuter}, we cannot recurse immediately unless more than three gates have been removed (as in step 28). It is possible that the circuit after simplifying, $\cC'$, correctly computes $(\neg)\XOR_{n-1}$; $\cC$ might only err when $x_p$ is set the other way.

    At step 31 \textsc{XOR-Checker}, checks whether this is the case and if it is not then it can recurse. Since $\cC'$ does not compute the appropriate parity on $I \setminus \{p\}$, the returned value is also incorrect for $\cC$. If the algorithm does not return, then \textsc{XOR-Checker} must restrict $x_p$ in the opposite way from Schnorr to find an error. However, this eliminates $\beta$ via a \emph{passing} rule and does not guarantee $\gamma$ is removed. 
    
    If $\cC$ is structured as in Figure \ref{fig:checker-corner-case-1} (negations are omitted) it is possible for only two gates to be removed (e.g. if $\beta$ and $\alpha$ are both removed via a passing rule). We can however instead argue that $\cC$ is non-optimal and reduce its circuit size. Notice the subcircuit rooted at $\beta$ computes a binary Boolean function of $x_p$ and $x_q$ using two costly gates. This circuit is too small to compute $(\neg)\XOR_2$, and thus must compute a function in $\mathbf{U}_2$. However, all such functions can be computed using at most one binary gate in the DeMorgan basis. Therefore \textsc{XOR-Checker} can use brute force to determine which function $\beta$ computes, replace $\beta$ with an optimal circuit for that function, and remove $\alpha$. This yields a circuit of size strictly less than $3(k-1)$, and the algorithm can call \textsc{XOR-Refuter} to find an input it (and by extension $\cC$) err on.

    Otherwise restricting $x_p$ the opposite way is also guaranteed to remove at least three binary gates. Furthermore, the resulting circuit does not compute the appropriate parity of the inputs indexed by $I \setminus \{p\}$ and hence $\textsc{XOR-Refuter}$ returns an input, consistent with $\vec{a}$ on $p$ and non-$I$ indices, on which the simplified circuit, and by extension $\cC$ err.

    As detecting optimal $\XOR_n$ circuits can be done in polynomial time and \textsc{XOR-Refuter} is a polynomial time refuter, Algorithm \ref{alg:xor-checker} also runs in polynomial time. \qedhere

        \begin{figure}[t]
            \centering
            \includegraphics[scale=1]{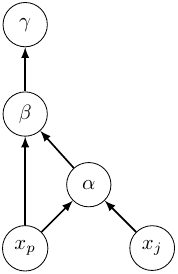}
            \caption{The local structure around $x_p$ where $x_p \gets 1 - \vec{a}_p$ might only eliminate two gates.}
            \label{fig:checker-corner-case-1}
        \end{figure} 
\end{proof}
\section{Subroutines for Affine Refuter}
\label{sec:affine-subroutines}

In this section we define Algorithms $\ref{alg:parity-support} - \ref{alg:find-substitution}$ used by Algorithm \ref{alg:affine-refuter}. For completeness, we show that these algorithms are correct and run in polynomial time.

\paragraph*{ParitySupport}

\begin{algorithm}[h]
    \caption{Identifies the affine function by a circuit of fanout $1$ $\oplus$-type gates}\label{alg:parity-support}
    \begin{algorithmic}[1]
    \Require $C$, a normalized circuit of fanout $1$ $\oplus$-type gates
    \Ensure $I, c$ such that $C$ computes $\bigoplus_{i \in I} x_i \oplus c$
    \Procedure{ParitySupport}{C}
    \State $I \gets \emptyset$
    \For{$i \in [n]$}
        \If{\Call{fanout}{$x_i$} is odd}
            \State add $i$ to $I$
        \EndIf
    \EndFor
    \State \Return $I,C(\vec{0})$
    \EndProcedure
    \end{algorithmic}
\end{algorithm}

\begin{proof}
    Let $\cC$ be a normalized circuit whose only gates (if any) are fanout $1$ $\oplus$-type gates. Then $C$ computes the formula where $\oplus$ is the operator and the inputs are variables and constants. As $\oplus$ commutes, it is easy to see that every variable which appears in the formula an even number of times is canceled out. Thus $\cC$ computes $\bigoplus_{i \in I} x_i \oplus c$ for some $c \in \{0,1\}$. Furthermore, $\cC$, $\cC(\vec{0}) = \bigoplus_{i \in I} 0 \oplus c = c$ as desired. Checking the fanout of each $x_i$ and evaluating $\cC$ on one input can be done in quadratic time with respect to $|\cC|$.
\end{proof}

\paragraph*{AffineIntersect}

\begin{algorithm}[h]
    \caption{Computes the description for an affine space and linear restriction}\label{alg:affine-intersection}
    \begin{algorithmic}[1]
    \Require $U \in \F_2^{n \times d'}, \vec{u} \in {\F_2^n}, V \in \F_2^{n \times (n-1)},$ and $\vec{v} \in  \F_2^n$ where $U$ and $V$ are full column rank and $\{U\vec{x} + \vec{u}\} \cap \{V\vec{y} + \vec{v}\} \neq \emptyset$
    \Ensure $W, \vec{w}$ such that $W \in \F_2^{n \times \delta}$ for $\delta \in \{d'-1, d\}$ is full column rank and $\{W\vec{z} + \vec{w}\} = \{U\vec{x} + \vec{u}\} \cap \{V\vec{y} + \vec{v}\}$
    \Procedure{AffineIntersect}{$U,\vec{u},V,\vec{v}$}
    \State $\mat{\vec{x}_0 \\ \vec{y}_0} \gets$ a solution to $[U \mid V] \mat{\vec{x} \\ \vec{y}} = \vec{u} + \vec{v}$
    \State $\mat{\vec{x}_1 \\ \vec{y}_1}, \ldots, \mat{\vec{x}_\delta \\ \vec{y}_\delta} \leftarrow$ a basis for $\ker\left(\left[U \mid V\right]\right)$
    \State $X \leftarrow $ a $D \times \delta$ matrix whose $\delta$ columns are $x_1, x_2, \ldots, x_\delta$
    \State \Return $UX, U\vec{x}_0 + \vec{u}$
    \EndProcedure
    \end{algorithmic}
\end{algorithm}

\begin{proof}
    Let $U \in \F_2^{n \times d'}, \vec{u} \in \F_2^n, V \in \F_2^{n \times (n-1)}$, and $\vec{v} \in \F_2^n$ where $U$ and $V$ are full column rank and $\{Ux + \vec{u} \cap \{Vy + \vec{y}\} \neq \emptyset$. We first argue that steps 2 and 3 can be computed. As the intersection is not empty, there exists some $\vec{x}', \vec{y}'$ such that $U\vec{x}' + \vec{u} = V\vec{y}' + \vec{v}$. Rearranging yields that $U\vec{x}' + V\vec{y}' = \vec{u} + \vec{v}$. We rewrite this using block matrices to get $[U \mid V] \mat{\vec{x}' \\ \vec{y}'} = \vec{u} + \vec{v}$. Since there is a solution to this system, we can perform Gaussian elimination to find one. Furthermore, Gaussian elimination can also be used in step $3$ to find a basis for $\ker([U\mid V])$. 
    
    We show that $\delta$ is either $d'$ or $d'-1$. Let $\delta$ be $\dim(\ker([U \mid V]))$ and observe that $\dim(\ker([U | V]) = \dim(\im(U) \cap \im(V)) =  \dim(U) + \dim(V) - \dim(\im(U) \cup \im(V))$. As $\dim(U) = d'$, $\dim(V) = n-1$ and $\dim(\im(V)) \leq \dim(\im(U) \cup \im(V)) \leq \dim(F_2^n) = n$, we conclude that $\delta$ is either $d$ or $d-1$.

    We verify that $W = UX$ and $\vec{w} = Ux_0 + \vec{u}$ define the intersection of the two affine spaces. Fix $\hat{z} \in \F_2^\delta$ and consider $W\vec{z} + \vec{w}$. Unwrapping definitions, we have $W\vec{z} + \vec{w} = UX\vec{z} + U\vec{x}_0 + \vec{u} = U\left(X\vec{z} + \vec{x}_0\right) + \vec{u}$ and thus by definition it is in $\{U\vec{x} + \vec{u}\}$. To show $W\vec{z} + \vec{w} \in \{V\vec{y} + \vec{v}\}$, we must first rewrite $X$. 
    
    Let $K$ be the matrix whose columns are the computed basis for $\ker([A | B])$. Then $X = [\mathbf{I}_\delta | \mathbf{0}_{n -1}]K$ where $\mathbf{I}_\delta$ is the $\delta \times \delta$ identity matrix and $\mathbf{0}_{n-1}$ is the $(n-1) \times (n-1)$ zero matrix. We also rewrite $[U \mid V] = U[\mathbf{I}_\delta \mid \mathbf{0}_{n -1}] + V[\mathbf{0}_\delta \mid \mathbf{I}_{n-1}]$. As $[U \mid V]Kv = \vec{0}$ for any $v$, we have $U[\mathbf{I}_\delta \mid \mathbf{0}_{n -1}]K\vec{z} = V[\mathbf{0}_\delta \mid \mathbf{I}_{n-1}]K\vec{z}$. Furthermore, $U\vec{x}_0 + \vec{u} = V\vec{y_0} + \vec{v}$ since $\mat{\vec{x}_0 \\ \vec{y}_0}$ as a solution. Hence,
        \begin{align*}
        W\vec{z} + \vec{w} &= UX\vec{z} + U\vec{x}_0 + \vec{u} \\ 
        &= U[\mathbf{I}_\delta \mid \mathbf{0}_{n-1}]K\vec{z} + U\vec{x}_0 + \vec{u} \\
        &=V[\mathbf{0}_\delta \mid \mathbf{I}_{n-1}]K\vec{z} +  V\vec{y}_0 + \vec{v} \\
        &=V\left([\mathbf{0}_\delta \mid \mathbf{I}_{n-1}]K\vec{z} +  \vec{y}_0\right) + \vec{v} \\
    \end{align*}

    We now show that every element in the intersection is also in $\{W\vec{z} + \vec{w}\}$. Let $\vec{u}$ be an element of the intersection: $u = U\vec{x} + \vec{u} = V\vec{y} + \vec{v}$ for some $\vec{x}, \vec{y}$. Rearranging yields $U\vec{x} + V\vec{y} = \vec{u} + \vec{y}$ and therefore $[U | V] \mat{\vec{x} \\ \vec{y}} = \vec{u} + \vec{v}$. As $\mat{\vec{x} \\ \vec{y}}$ is a solution to the system, we can write that $\mat{\vec{x} \\ \vec{y}} = \mat{\vec{x}' \\ \vec{y}'} + \mat{\vec{x}_0 \\ \vec{y}_0}$ for some $\mat{\vec{x}' \\ \vec{y}'} \in \ker\left([U | V]\right).$ Therefore $\mat{\vec{x}' \\ \vec{y}'} = K\vec{z}'$ for some $\vec{z}' \in \F_2^\delta$. Therefore:
    $$ u = U\vec{x} + \vec{u} = U[\mathbf{I}_\delta \mid \mathbf{0}_{n-1}]\left(K\vec{z}' + \mat{\vec{w}_0 \\ \vec{z}_0}\right) + \vec{u} = UX\vec{z'} + U\vec{w}_0 + \vec{u} = W\vec{z}' + \vec{w}.$$

    Therefore the algorithm correctly outputs the direction matrix and offset for the intersection of the given affine subspaces. As Gaussian elimination is efficient, the overall algorithm runs in polynomial time with respect to $n$ as desired.
\end{proof}

\paragraph*{ConstraintToAffine}

\begin{algorithm}[h]
    \caption{Computes the affine space corresponding to a parity constraint}\label{alg:constraint-to-affine}
    \begin{algorithmic}[1]
    \Require $I \subseteq [n], c \in \{0,1\}$ such that $I \neq \emptyset$
    \Ensure $B, \vec{v}$ such $B \in \F_2^{n \times (n-1)}, b \in \F_2^n$ and $\{B\vec{z} + b\} = \{\vec{x} \in F_2^n \mid \bigoplus_{i \in I} \vec{x}_i = c\}$
    \Procedure{ConstraintToAffine}{$I,c$}
    \State $\vec{z}_0 \leftarrow$ a solution to $\ind{I} \vec{z} = c$
    \State $\vec{z}^1, \vec{z}^2, \ldots z^{n-1} \gets$ a basis for $\ker(\ind{I})$
    \State $Z \gets$ the $n \times (n-1)$ matrix whose columns are $\vec{z}^1, \vec{z}^2, \ldots z^{n-1}$
    \State \Return $Z, \vec{z}_0$
    \EndProcedure
    \end{algorithmic}
\end{algorithm}

\begin{proof}
    The algorithm is well-defined. As $I \neq 0$, there is a solution to $\ind{I}(\vec{z}) = c$. The algorithm can solve the system and find a solution and find a basis for the kernel. As the rank of $\ind{I}$ is $1$, the basis has dimension $n-1$ by the rank-nullity theorem. Correctness follows immediately: $\bigoplus_{i \in I} x_i = c$ is a linear equation over $\F_2$ and its solution set is exactly the affine space $\vec{z}_0 + \ker(\ind{I})$.

    As Gaussian elimination is efficient and can be used to implement sets $2$ and $3$, the algorithm runs in polynomial time with respect to $n$.
\end{proof}

\paragraph*{FindSubstitution}

\begin{algorithm}[h]
    \caption{Computes a substitution which makes $\cC$ output constant $b$}\label{alg:find-substitution}
    \begin{algorithmic}[1]
    \Require $\cC$ a normalized non-constant circuit of fan-out $1$ $\oplus$ gates, $b \in \{0,1\}$
    \Ensure $\cS, j$ such that $\cS \equiv \bigoplus_{i \in I} x_i \oplus c_0$ for some $I$ where substituting $x_j$ in $C$ for $\cS$ makes $C$ output constant $b$
    \Procedure{FindSubstitution}{$\cC,b$}
    \State $(I, c) \leftarrow$ \Call{ParitySupport}{$\cC$}
    \State $j \leftarrow$ any element of $I$
    \State $\cS \leftarrow \bigoplus_{i \in I \setminus \{j\}} x_i \oplus (c \oplus b)$ as a circuit
    \State \Return $\cS, j$
    \EndProcedure
    \end{algorithmic}
\end{algorithm}

\begin{proof}
    The algorithm is correct. As $\cC$ is non-constant and $\textsc{paritysupport}$ is correct, $\cC$ computes $\bigoplus_{i \in I} x_i \oplus c$ and $I \neq \emptyset$. Substituting $x_j \gets S$ in $\cC$ means that $\cC$ computes $\bigoplus_{i \in I \setminus \{I\}} x_i \oplus \left(\bigoplus_{i \in I \setminus \{I\}} x_i \oplus (c \oplus b)\right)  \oplus  c = b$ as desired. As \textsc{paritysupport} runs in polynomial time, and $|S| = O(|\cC|)$, the algorithm runs in polynomial time.
\end{proof}

\end{document}